\documentclass[10pt]{amsart}

\usepackage{amssymb,amsthm,amsmath}
\usepackage[numbers,sort&compress]{natbib}
\usepackage{color}
\usepackage{graphicx}

\title[ ]{The     spectra  of   surface Maryland model for all phases}

\author{Wencai Liu}
\address[Wencai Liu]{School of Mathematical Sciences, Fudan University, Shanghai 200433, P. R. China} \email{liuwencai1226@gmail.com}
\address{Current address: Department of Mathematics, University of California, Irvine, California 92697-3875, USA}

\theoremstyle{plain}
\newtheorem{theorem}{Theorem}[section]

\newtheorem{lemma}[theorem]{Lemma}

\theoremstyle{definition}
\newtheorem{definition}[theorem]{Definition}

\newtheorem{remark}[theorem]{Remark}

\begin{document}


\begin{abstract}
We study the discrete Schr\"{o}dinger operators $H_{\lambda,\alpha,\theta}$ on $\ell^2(\mathbb{Z}^{d+1})$ with surface potential of the form $V(n,x)=\lambda \delta(x)\tan\pi(\alpha\cdot n+\theta)$,  and  $H_{\lambda,\alpha,\theta}^{+}$ on $\ell^2(\mathbb{Z}^{d}\times \mathbb{Z}_+)$ with the boundary condition $ \psi_{(n,-1)}=\lambda \tan\pi(\alpha\cdot n+\theta)\psi_{(n,0)} $, where $\alpha\in \mathbb{R}^d$  is rationally independent. We show that the spectra of $H_{\lambda,\alpha,\theta}$ and  $H_{\lambda,\alpha,\theta}^{+}$  are $(-\infty,\infty)$
for all parameters.  
We can also  determine  the absolutely  continuous spectra  and Hausdorff dimension  of the spectral measures   if $d=1$.

\end{abstract}
\maketitle

 \section{Introduction}

    In this paper, we study the discrete Schr\"{o}dinger operators with surface potential. Let $ d\geq 1$ be given and $ \mathbb{Z}_+=\{0,1,2, \cdots\}$. We denote the points in $\mathbb{Z}^{d+1} $ (or  $\mathbb{Z}^{d+1}_+=\mathbb{Z}^{d}\times \mathbb{Z}_+$) by $(n,x)$, where $n\in  \mathbb{Z}^{d} ,x\in \mathbb{Z}$
    (or $n\in \mathbb{Z}^{d}$, $x\in  \mathbb{Z}_+$).
    \par
    Given $y=(y_1,y_2,\cdots,y_k)\in \mathbb{Z}^{k}$ with $k\geq 1$, denote by $||y||_k=\sum_{j=1}^k|y_j|$.

    Now we introduce two self-adjoint operators $H_{\lambda,\alpha,\theta}$ and $H_{\lambda,\alpha,\theta}^+$. The operator $H_{\lambda,\alpha,\theta}$ acts on $\ell^2(\mathbb{Z}^{d+1} )$ as
    \begin{equation}\label{Def.H}
        ( H_{\lambda,\alpha,\theta}\psi)_{(n,x)}=\sum_{||(n-n^\prime,x-x^\prime) ||_{d+1}=1}\psi_{(n',x')} +\lambda \delta(x)\tan\pi(\alpha\cdot n+\theta)\psi_{(n,x)},
    \end{equation}
 where $\lambda\neq 0$ is the coupling, $\alpha\in \mathbb{R}^d$ is the frequency, and $\theta $ is the phase. In order to make the potential well defined,
we  always assume $\theta$  satisfies the following condition
\begin{equation}\label{Contheta}
   k\cdot \alpha +\theta \neq \frac{1}{2} \mod \mathbb{Z} \text{ for any } k\in \mathbb{Z}^d.
\end{equation}
  The other operator $H_{\lambda,\alpha,\theta}^+$ on $\ell^2( \mathbb{Z}^{d}\times \mathbb{Z}_+ )$ is given by the following equation,
     \begin{equation}\label{Def.H^+}
        ( H_{\lambda,\alpha,\theta}^+\psi)_{(n,x)}=
\left\{
  \begin{array}{ll}
   \sum_{||(n-n^\prime,x-x^\prime) ||_{d+1}=1}\psi_{(n',x')}  , & \hbox{if   $x>$ 0;} \\
   \psi_{(n,1)} + \sum_{||n'-n||_d =1}  \psi_{(n',0)  } +\lambda \tan\pi(\alpha\cdot n+\theta)\psi_{(n,0)}, & \hbox{if $x=$ 0.}
  \end{array}
\right.
    \end{equation}

Whenever the meaning is clear within the context, we will drop the dependence on   parameters   $\lambda$, $\alpha$, $\theta$ and $d$. Thus
we write $H$ or $H_\theta$ for $H_{\lambda,\alpha,\theta}$,  and  $||n'-n||$  for $||n'-n||_d$.
\par
Recall that   the   Maryland model is given by
 \begin{equation*}
        ( h_ {\lambda,\alpha,\theta}\psi)_n=\sum_{||n'-n|| =1}\psi_{n' }+\lambda  \tan\pi(\alpha\cdot n+\theta)\psi_n,
    \end{equation*}
with $n, n'\in \mathbb{Z}^d$. It has been heavily studied in \cite{MR767188,grempel1982localization,MR776654,cyconschrodinger}.
In particular,
for $d=1$, Jitomirskaya-Liu has obtained the complete description of the three spectral types of Maryland model \cite{liulana1}.\color{black}
\par
 Now we want  to say more about  the surface multidimensional operators. Let $H_0$   be the discrete  Laplacian   on $\ell^2(\mathbb{Z}^{d+1})$,
and
$H_0^+$   be the discrete  Laplacian on $\ell^2(\mathbb{Z}^{d+1}_+)$ with Dirichlet boundary condition.
More precisely,
 \begin{equation*}
        ( H_0\psi)_y=\sum_{||y'-y|| =1}\psi_{y'},
    \end{equation*}
where $y,y'\in \mathbb{Z}^{d+1}$.
  And
    \begin{equation*}
        ( H_0^+\psi)_{(n,x)}= \sum_{||(n'-x',n-x)||=1}\psi_{(n',x')}
    \end{equation*}
 with $\psi_{(n,-1)}=0 $, where $n\in \mathbb{Z}^{d}$ and $x\in \mathbb{Z}_+$.
\par
Consider the operator $H$.
The  support of the  potential   is on the hyperplane $\mathbb{Z}^d$  of the space $\mathbb{Z}^{d+1}$. For the other one  $H^+$,
it is the discrete  Laplacian on $\ell^2(\mathbb{Z}_+^{d+1})$ with the boundary condition $\psi_{(n,-1)}=\lambda  \tan\pi(\alpha\cdot n+\theta)\psi_{(n,0 )}$.
Thus it is natural to call both  operators $H$ and  $H_+$ the surface Maryland model.
\par
 Denote
   by  
 $\sigma(H)$ and  $\sigma_{ac}(H)$  the  spectrum, and absolutely continuous spectrum of  the self-adjoint operator $H$.
We recall that $\sigma(H_0)=[-2(d+1),2(d+1)]$ ($\sigma(H_0^+)=[-2(d+1),2(d+1)]$)  and the spectrum of $H_0 $ ($H_0^+$) is purely absolutely
continuous. Denote by $c_d=2(d+1)$. \color{black}

We recall that $\alpha=(\alpha_1,\alpha_2,\cdots,\alpha_d)$ is rationally  independent if
for any choice of integers $(k_1,k_2,\cdots k_d)\ \neq 0$,
\begin{equation*}
    \sum_{j=1}^{d}k_j\alpha_j\notin \mathbb{Z}.
\end{equation*}
In the rest of the paper, we always assume $\lambda\neq 0$,  $\alpha=(\alpha_1,\alpha_2,\cdots,\alpha_d)$ is rationally  independent and
$\theta$ satisfies condition (\ref{Contheta}).  We will often use ``all
$\lambda (\alpha,\theta)$'', meaning ``all $\lambda (\alpha,\theta)$ as above''.

Now we start to state some known results.
\begin{theorem}\label{acsp}\cite{bentosela2008spectral,jakvsic1998spectrum}
Let $H_{\lambda,\alpha,\theta}$ and $H^+_{\lambda,\alpha,\theta}$ be given by (\ref{Def.H}) and  (\ref{Def.H^+}) respectively.
Then for all parameters $(\lambda ,\alpha,\theta)$,
$[-c_d,c_d]\subset \sigma_{ac}(H_{\lambda,\alpha,\theta}), \sigma_{ac}(H^{+}_{\lambda,\alpha,\theta})$, and
$H_{\lambda,\alpha,\theta}$,  $H^{+}_{\lambda,\alpha,\theta}$ has purely    absolutely
continuous spectrum on the interval $(-c_d,c_d)$.
\end{theorem}
\par
 We say $\alpha  \in \mathbb{R}^d $ satisfies a Diophantine condition if there exists  $\kappa>0$ and $\tau>0$,
such that
 \begin{equation}\label{Def.DC}
  ||k\cdot\alpha||_{\mathbb{R}/\mathbb{Z}}>\kappa ||k||^{-\tau}  \text{ for   any } k\in \mathbb{Z}^d\backslash \{0\},
 \end{equation}
where $ ||x||_{\mathbb{R}/\mathbb{Z}}=\min_j|x-j|$.

\par
It is to be believed that $H$ and $H^+$ have purely singular spectra outside the set $[-c_d,c_d]$.
Khoruzenko and Pastur \cite{khoruzhenko1997localization} have proved this
for $\alpha$ satisfying Diophantine condition.
\begin{theorem}\label{ThKP}\cite{khoruzhenko1997localization}
Assume frequency $\alpha$ satisfies Diophantine condition (\ref{Def.DC}). Then $\sigma(H_{\lambda,\alpha,\theta})=(-\infty,\infty)$ and
 $H_{\lambda,\alpha,\theta}$ only  has pure point on $(-\infty,\infty)\backslash[-c_d,c_d]$.
 On this set, the eigenvalues are simple and the
corresponding eigenfunctions decay exponentially.
 These  properties  also apply  to the
operator  $H_{\lambda,\alpha,\theta}^+$.
\end{theorem}


\par
For the surface Maryland model, one of the basic ideas
is to integrate the $x-$variable
and to reduce the $d+1-$dimensional spectral problem to a $d$-dimensional problem. This    technology has been heavily developed by use of
Green function in \cite{jaksic1998surface,jakvsic1998spectrum,jakvsic1999localization,khoruzhenko1997localization}.
This  technology can also be extended to
the generalized surface Maryland model
\cite{bentosela2008spectral,bentosela2005spectral}. In this paper, we also use this basic  technology to analyze the spectra of  surface Maryland model.
However, we  bring a new lighting on the subject.

Our first result is to obtain the explicit formula of  the spectra of $H_{\lambda,\alpha,\theta}$ and   $H_{\lambda,\alpha,\theta}^+$.
\begin{theorem}\label{Thspectra}
Suppose $\alpha$ is rationally independent and $\lambda\neq 0$, then
$\sigma (H_{\lambda,\alpha,\theta})=\sigma(H_{\lambda,\alpha,\theta}^+)=(-\infty,+\infty)$  for all $\theta$.
\end{theorem}
{\bf Remark}
From Theorems \ref{acsp} and \ref{ThKP}, we directly  know that Theorem \ref{Thspectra} holds for  Diophantine $\alpha$.
Diophantine condition is necessary for Theorem \ref{ThKP} when we solve the cohomological equation.
Our aim is to  prove the   Theorem \ref{Thspectra}   for  non-Diophantine $\alpha$.

Secondly we study the  absolutely continuous spectrum.  By  Theorem \ref{acsp},  $H$ ($H^+$) has purely    absolutely
continuous spectrum in  the interval $(-c_d,c_d)$, it was  not known if  $H$ ($H^+$) has some    absolutely continuous component
 in the remaining interval $\mathbb{R}\backslash [-c_d,c_d]$. It is believed that $H$ ($H^+$) has no     absolutely
continuous component  in $\mathbb{R}\backslash [-c_d,c_d]$.
  In this paper, we solve the problem for $d=1$. We construct an arithmetically defined measure zero set (actually  Hausdorff dimension zero) that supports
spectral  measures in the regime $\mathbb{R}\backslash [-c_d,c_d]$, which directly implies the following Theorem.
\begin{theorem}\label{Thac}
Let $H_{\lambda,\alpha,\theta}$ and   $H_{\lambda,\alpha,\theta}^+$ be given by (\ref{Def.H}) and (\ref{Def.H^+}) respectively with $d=1$.
Suppose $\alpha\in \mathbb{R}\backslash \mathbb{Q}$, then
$\sigma_{ac} (H_{\lambda,\alpha,\theta})=\sigma_{ac}(H_{\lambda,\alpha,\theta}^+)=[-4,4]$
and
$H_{\lambda,\alpha,\theta}$,  $H_{\lambda,\alpha,\theta}^+$ has purely absolutely continuous spectra on $(-4,4) $ for all $\theta$.
Furthermore, the spectral measures of   both operators are of  Hausdorff dimension one restricted  to $(-4,4)$ and  of Hausdorff dimension zero restricted to $\mathbb{R}\backslash [-4,4]$ for all $\theta$.

\end{theorem}
\begin{remark}\label{revisere}
The fact
that the  spectral measures restricted  to $(-4,4)$ are of Hausdorff dimension one  is trivial because $H_{\lambda,\alpha,\theta}$ and $H_{\lambda,\alpha,\theta}^+$ have purely absolutely continuous spectra on $(-4,4) $.
\end{remark}
 We will  prove  Theorems \ref{Thspectra},\ref{Thac}  in two steps. Firstly,   we will reduce the $d+1$-dimensional spectral problem  to a $d$-dimensional problem in \S 2. Secondly, we will study the $d$-dimensional problem  using  the ideas in paper \cite{liulana1} (\S3 and \S4), in which the authors deal with the Maryland model.


\section{Green   function and Dimension Reduction}
Because our arguments are almost the same for the operators $H$ and $H^+$,   we concentrate on $H$ in  the following discussion,  and only point out the difference for $H^+$.

In this section, we carry out the dimension reduction for the surface Maryland model, and lay the ground for our main results.
We first consider the operator $H$.  We should point out that all the contents in this section   are   from \cite{bentosela2008spectral,jakvsic1998spectrum,khoruzhenko1997localization}. For convenience, we rewrote it here.
\par
Denote by $G(X_1;X_2)$ and $G_0(X_1-X_2)$ the Green  functions of the operators $H$ and $H_0$ respectively, i.e., the infinite dimensional matrices
$(H-zI)^{-1}$ and $(H_0-zI)^{-1}$ with $\Im z\neq 0$.  Then, using  the resolvent identity,
one has
\begin{equation}\label{Resolventidentity}
     G(X_1;X_2)=G_0(X_1-X_2)-\sum_{\eta\in \mathbb{Z}^d}G_0(n_1-\eta,x_1)\lambda v(\eta) G(\eta,0;X_2),
\end{equation}
where $X_1=(n_1,x_1)\in \mathbb{Z}^{d+1}$, $X_2=(n_2,x_2)\in \mathbb{Z}^{d+1}$ and $v(\eta)=  \tan\pi(\alpha\cdot \eta+\theta)$.
For simpleness, denote by $v$ ($v^{-1}$) be the  multiplication operator by $v(\eta)$ ($1/v(\eta)$) on $\ell^2(\mathbb{Z}^d)$.
\par
Set $x_1=0$ in (\ref{Resolventidentity}), we find that  the function $ g_{x_2}(n_1;n_2)=G(n_1,0;n_2,x_2)$ satisfies
\begin{equation}\label{Resolventidentity1}
     g_{x_2}(n_1;n_2)=g_{0,x_2}(n_1-n_2)-\sum_{\eta\in \mathbb{Z}^d}\Gamma_0(n_1-\eta )\lambda v(\eta) g_{x_2}(\eta;n_2),
\end{equation}
where $ g_{0,x_2}(\eta)=G_0(\eta,-x_2)$ and $ \Gamma_0(\eta)=G_0( \eta,0)$ with $\eta\in \mathbb{Z}^d$. Next we always regard  $\Gamma_0(\eta)$ as a operator
on $\mathbb{Z}^d$.
\par
Fix $x_2$,  equation (\ref{Resolventidentity1}) becomes a problem in dimension $d$, and has a symbolic
solution $g_{x_2}=(I+\Gamma_0\lambda v)^{-1}g_{0,x_2} $. Combining with (\ref{Resolventidentity}), we have
\begin{eqnarray}
\nonumber
   G(X_1;X_2) &=&  G(n_1,x_1;n_2,x_2)\\
    &=&  G_0(X_1-X_2)-\sum_{\eta_1,\eta_2\in \mathbb{Z}^d} G_0(n_1-\eta_1,x_1)T(\eta_1,\eta_2)G_0(\eta_2-n_2,-x_2),\label{Resolventidentity2}
\end{eqnarray}
where $T=\lambda v(I+\Gamma_0\lambda v)^{-1}=(\lambda ^{-1}v^{-1}+\Gamma_0)^{-1}$.
\par
Notice that  the above derivation of equation (\ref{Resolventidentity2})  is only symbolic, because generally, the  operator
  $(I+\Gamma_0\lambda v)^{-1}$
is  not well defined. However, the following theorem shows (\ref{Resolventidentity2}) is true in fact.
 \begin{theorem}\label{Thm.Greenfunction1}(p.116-p.118, \cite{khoruzhenko1997localization})
Let $U$  denote the multiplication by $e^{2\pi i \alpha \cdot n}$ on $\ell^2(\mathbb{Z}^d)$ and $\delta=e^{2\pi i\theta}$.
If $\Im z \cdot \lambda<0$,
then equation (\ref{Resolventidentity2}) holds, where  $T$ can be written as
\begin{equation}\label{Def.T}
    T=\lambda (1-\delta U) (I-\delta CU)^{-1}(\lambda \Gamma_0-i I)^{-1},
\end{equation}
where
\begin{equation*}
    C=(\lambda \Gamma_0-i I)^{-1}(\lambda  \Gamma_0+i I).
\end{equation*}

\end{theorem}
\begin{proof}
Suppose we can prove $\lambda \Gamma_0-i I$ is invertible and $\|C\|<1$, then equation (\ref{Def.T}) is well defined, and (\ref{Resolventidentity2}) holds by directly computed.
Now we start to prove the facts  by using Fourier transformation.
Given a vector $f_n$ in $\ell^2(\mathbb{Z}^b)$ with $b=d$ or $b=d+1$, define a $L^2(\mathbb{T}^b)$ function  by the following equation,
\begin{equation*}
     \hat{f}(y)=\sum_{n\in \mathbb{Z}^b}f_ne^{ -2\pi i n\cdot y}.
\end{equation*}
Recall that $\mathbb{T}=\mathbb{R}/\mathbb{Z}$.
We call $\hat{f}(y)$ be the momentum representation of vector $ f_n$.
Under the Fourier transformation, $H_0$ become a multiplication operator on $L^2(\mathbb{T}^{d+1})$ and denote by $\hat{H}_0$.
More precisely,
\begin{equation*}
    ( \hat{H}_0 \hat{f})(y)=(\sum_{j=1}^{d+1}2\cos2\pi y_j)\hat{f}(y) .
\end{equation*}
Thus in momentum representation, Green function $G_0$ has the following simple form,
\begin{equation*}
    ( \hat{G}_0 \hat{f})(y)=(\frac{1}{\sum_{j=1}^{d+1}2\cos2\pi y_j-z})\hat{f}(y) ,
\end{equation*}
with $y\in  \mathbb{T}^{d+1}$.
And
 the operator $\Gamma_0$ becomes
\begin{equation*}
    ( \hat{\Gamma}_0 \hat{f})(y)= \hat{\Gamma}_0(y) \hat{f}(y) ,
\end{equation*}
where
\begin{equation}\label{Def.Gamma_0}
    \hat{\Gamma}_0(y)=\int_{\mathbb{T}}\frac{1}{\sum_{j=1}^{d+1}2\cos2\pi y_j-z}dy_{d+1} ,
\end{equation}
with
$y\in  \mathbb{T}^{d}$.
Notice that $ \Im  \lambda \hat{\Gamma}_0(y)<0$, then  $\lambda \Gamma_0-i I$ is invertible(because of $ \Im z\cdot \lambda<0$).
Furthermore,  $C$  also becomes a   multiplication operator in momentum representation,
\begin{equation*}
   ( \hat{C}\hat{f}) (y)= \frac{\lambda \hat{\Gamma}_0(y)+i }{\lambda \hat{\Gamma}_0(y)-i}\hat{f}(y).
\end{equation*}
This yields   $\|C\|<1$.
\end{proof}
Now we turn to the operator $H^+$. If a operator $A$  generates  from $H $,  then denote by   $A^+$ be the corresponding operator generated from $H^+$.
For example $G^+= (H^+-zI)^{-1} $.
Using a similar arguments like before,
one has
\begin{equation}\label{Resolventidentity+}
     G^+(X_1;X_2)=G_0^+(X_1;X_2)-\sum_{\eta\in \mathbb{Z}^d}G_0^+(n_1,x_1; \eta,0)v(\eta) G^+(\eta,0;X_2),
\end{equation}
where $X_1=(n_1,x_1)\in \mathbb{Z}_+^{d+1}$, $X_2=(n_2,x_2)\in \mathbb{Z}_+^{d+1}$ and $v(\eta)=\lambda  \tan\pi(\alpha\cdot \eta+\theta)$.
We should point out that  $G^+_0(n_1,x_1;n_2,x_2)$ depends  on $n_1-n_2,x_1,x_2$, this is different from $G_0(n_1,x_1;n_2,x_2)$, which only
depends on $n_1-n_2,x_1-x_2$.
\par
Set $x_1=0$ in (\ref{Resolventidentity+}), we find that  the function $ g_{x_2}^+(n_1;n_2)=G^+(n_1,0;n_2,x_2)$ satisfies
\begin{equation}\label{Resolventidentity1+}
     g_{x_2}^+(n_1;n_2)=g_{0,x_2}^+(n_1; n_2)-\sum_{\eta\in \mathbb{Z}^d}\Gamma_0^+(n_1-\eta )v(\eta) g_{x_2}^+(\eta;n_2),
\end{equation}
where $ g_{0,x_2}^+(n_1;n_2)=G_0^+(n_1,0;n_2,x_2)$ and $ \Gamma_0^+(\eta)=G_0^+(\eta,0; 0,0)$ with $\eta\in \mathbb{Z}^d$.
\par
Fix $x_2$,  equation (\ref{Resolventidentity1+}) becomes a problem in dimension $d$, and has a symbolic
solution $g_{x_2}^+=(I+\Gamma_0^+v)^{-1}g_{0,x_2}^+ $. Combining with (\ref{Resolventidentity+}), we have
\begin{eqnarray}
\nonumber
   G^+(X_1;X_2) &=&  G^+(n_1,x_1;n_2,x_2)\\
    &=&   G^+_0(n_1,x_1;n_2,x_2)-\sum_{\eta_1,\eta_2\in \mathbb{Z}^d} G_0^+(n_1,x_1;  \eta_1,0)T^+(\eta_1,\eta_2)G_0^+(\eta_2,0; n_2,x_2),\label{Resolventidentity2+}
\end{eqnarray}
where $T^+=v(I+\Gamma_0^+v)^{-1}=(v^{-1}+\Gamma_0^+)^{-1}$.
\par
We should say more about the operator $ \Gamma_0^+ $.
 In   momentum representation, $ \Gamma_0^+ $ has the following form   \cite{jakvsic1998spectrum},
\begin{equation*}
    (\hat{\Gamma}^+_0\hat{f})(y)= -r(y,z)\hat{f}(y) ,
\end{equation*}
 where $ r(y,z)$ is the root of the quadratic equation
\begin{equation*}
    X+\frac{1}{X}+\sum_{j}^d2\cos2\pi y_j=z
\end{equation*}
with $|r(y,z)|<1$.
\par
By a similar proof above, we have the following theorem.
 \begin{theorem}\label{Thm.Greenfunction2}
Suppose   $\Im z \cdot \lambda<0$,  then equation
(\ref{Resolventidentity2+})  holds and
   $T^{+}$ can    be written as
\begin{equation}\label{Def.T+}
    T^{+}=\lambda (1-\delta U) (I-\delta C^+U)^{-1}(\lambda \Gamma_0^+-i I)^{-1},
\end{equation}
where
\begin{equation*}
    C^+=(\lambda \Gamma_0^+-i I)^{-1}(\lambda  \Gamma_0^++i I).
\end{equation*}
\end{theorem}
\begin{remark}
Notice that for $z\in \mathbb{C}\backslash [-c_d,c_d]$, the operators  $\Gamma_0$  and $ \Gamma_0^+$  are well defined.
\end{remark}
\section{The spectra of surface Maryland model}
In this section, we will prove that  for the  surface  Maryland model $ \sigma(H_{\lambda,\alpha,\theta})=(-\infty,\infty)$  and
$ \sigma(H_{\lambda,\alpha,\theta}^+)=(-\infty,\infty)$.
From Theorem \ref{ThKP}, one has that for DC $\alpha$,  $ \sigma(H_{\lambda,\alpha,\theta})=(-\infty,\infty)$, $ \sigma(H_{\lambda,\alpha,\theta}^+)=(-\infty,\infty)$ for all $\lambda\neq 0$ and $\theta$.
Our main work in this section is to show that the spectrum $\sigma(H_{\lambda,\alpha,\theta})$ does not depend on  $\theta$.
\par
Following the discussion of unbounded operators in section VIII \cite{reed1972methods},  it suffices to show the self-adjoint operators  $H_{\lambda,\alpha,\theta}$
is continuous in norm resolvent sense with respect to $\theta$.  Combining with the specific formula of  the Green function of surface  Maryland model
(Theorem \ref{Thm.Greenfunction1} and Theorem \ref{Thm.Greenfunction2}),
we can obtain the results. In  \cite{liulana1}\color{black}, the authors have already  used this methods to show that the spectrum of  Maryland model does not depend on phase $\theta$.

 \begin{definition}\label{Def.Normresolvent}
 Let $B_n$, $n=1,2, \cdots$ and $B$ be self-adjoint operators, then $B_n$ is said to converge to $B$ in norm resolvent sense
 if $G_{z}(B_n)\rightarrow G_{z}( B)$ in norm for all $z$ with $\Im z\neq 0$, where $G_{z}(A)=(A-zI)^{-1}$.
 \end{definition}
\begin{theorem}(Theorem VIII.24, \cite{reed1972methods})\label{Continuousspectrum}
Let $B_n$, $n=1,2, \cdots$ and $B$ be self-adjoint operators, and  $B_n$  converge to $B$ in norm resolvent sense.
Then for any $E\in \sigma(B)$, there exists $E_n\in \sigma(B_n)$ such that $E_n\rightarrow E$.
\end{theorem}

\begin{theorem}\label{Spectrumforalltheta}
The spectrum of surface Maryland model does not depend on phase $\theta$. More precisely,
set
$\sigma(H_{\lambda,\alpha,\theta}) $ or  $\sigma(H_{\lambda,\alpha,\theta}^+) $ does not depend on   phase $\theta$.
\end{theorem}
\begin{proof}

For  given  two phases $\theta_1$ and $\theta_2$,   there exists some sequence $j_k\in \mathbb{Z}^d$
such that
\begin{equation*}
 \lim_{k\rightarrow \infty}||\theta_1 +j_k \cdot \alpha-\theta_2||_{\mathbb{R}/\mathbb{Z}}\rightarrow 0 ,
\end{equation*}
since $\alpha$ is rationally  independent.

 Let $B_k= H_{\lambda,\alpha,\theta_1+j_k\cdot \alpha}$ ($B_k= H_{\lambda,\alpha,\theta_1+j_k \cdot\alpha}^+$) and
$B=H_{\lambda,\alpha,\theta_2}$ ($B=H_{\lambda,\alpha,\theta_2}^+$). We must have $B_k$   converge to $B$ in norm resolvent sense.
Indeed, if $ \lambda\cdot  \Im z<0$, this is easy to see by   (\ref{Resolventidentity2}) and (\ref{Def.T}) ((\ref{Resolventidentity2+}) and (\ref{Def.T+})); if $\lambda \cdot\Im z>0$,  this is also true     by the fact $H_{\lambda,\alpha,\theta}=H_{-\lambda,-\alpha,-\theta}$ ($H_{\lambda,\alpha,\theta}^+=H_{-\lambda,-\alpha,-\theta}^+$).
   \par
   Applying Theorem \ref{Continuousspectrum}, for any $E\in \sigma(H_{\lambda,\alpha,\theta_2})$ ($\sigma(H_{\lambda,\alpha,\theta_2}^+)$), there exists
  $E_k\in \sigma(B_k)$ such that $E_k\rightarrow E$.  Clearly, the spectrum of operator  $B_k$ does not depend on $k$, i.e.,
  $\sigma(B_k)=\sigma(H_{\lambda,\alpha,\theta_1})$ ($\sigma(B_k)=\sigma(H_{\lambda,\alpha,\theta_1}^+)$) because the shift operator is a unitary operator. Thus
   we must have
   \begin{equation*}
     \sigma(H_{\lambda,\alpha,\theta_2})\subseteq  \sigma(H_{\lambda,\alpha,\theta_1}),
   \end{equation*}
and
\begin{equation*}
     \sigma(H_{\lambda,\alpha,\theta_2}^+)\subseteq  \sigma(H_{\lambda,\alpha,\theta_1}^+).
   \end{equation*}
This implies the  theorem since    $ \theta_1$ and $ \theta_2$ are arbitrary.
\end{proof}
{\bf The proof of Theorem \ref{Thspectra}}
\begin{proof}
To avoid repetition,
we only need to prove the theorem for   operator  $H$. Denote by $ \mathbb{E} \{\cdot\}$  the average value over $\mathbb{T}$.
Since $\|C\|<1$, one has
\begin{equation*}
   (I-\delta CU)^{-1}=I+\sum_{k=1}^{\infty} (\delta CU)^k.
\end{equation*}
Notice that $\delta U=U\delta$ and $\delta C=C\delta$, we have
\begin{equation*}
   (I-\delta CU)^{-1}=I+\sum_{k=1}^{\infty} \delta^k(CU)^k.
\end{equation*}
 Putting (\ref{Def.T})
into  (\ref{Resolventidentity2}) and integrating  over $\mathbb{T}$ with respect to variable $\theta$, one has that
$\mathbb{E}\{G\}$ does not depend on $\alpha$.
\par
Suppose   Theorem \ref{Thspectra} does not hold, then   by Theorem \ref{Spectrumforalltheta} there exist some $\alpha_1\in \mathbb{R}^d$ and  an interval $I$
such that $I\subset(-\infty,\infty)\backslash \sigma(H_{\lambda,\alpha_1,\theta})$ for all $\theta$.
Thus for any given vectors  $\delta_y$, $y\in \mathbb{Z}^{d+1}$, the
spectral measure $\mu _{\theta,\alpha_1, y}$ determined by
\begin{equation}\label{reequ1}
  (  ( H_{\lambda,\alpha_1,\theta}-zI)^{-1} \delta_{y}, \delta_{y})=\int \frac{d\mu_{\theta,\alpha_1,y}(\phi)}{\phi-z }
\end{equation}
satisfies
\begin{equation*}
    \mu_{\theta,\alpha_1,y}(I)=0 \text{ for all } \theta.
\end{equation*}
\par
Using the fact that
 $\mathbb{E}\{G\}$ does not depend on $\alpha$ and (\ref{reequ1}),
one has
\begin{equation*}
  \int_{\mathbb{T}}  \mu_{\theta,\alpha,y}(I)d\theta = \int_{\mathbb{T}}  \mu_{\theta,\alpha_1,y}(I)d\theta=0
\end{equation*}
for any  $\alpha\in\mathbb{R}^d$ and $y \in \mathbb{Z}^{d+1}$.
Combining  with the fact that $\sigma(H_{\lambda,\alpha,\theta})$ does not depend on $\theta$,   we have
 \begin{equation*}
 I\subset(-\infty,\infty)\backslash \sigma(H_{\lambda,\alpha,\theta})
 \end{equation*}
for all  $\alpha\in\mathbb{R}^d$.
This is contradicted to Theorems \ref{acsp} and  \ref{ThKP}.
\end{proof}
\section{ The spectral measures  for $d=1$}
  We always assume $E\in (-\infty,\infty)\backslash [-c_d,c_d]$ below. 
\par

Let
\begin{equation*}
    \mathcal{P}_{b}=\{\{\psi_y\}_{y\in \mathbb{Z}^{b} }: |\psi_y|\leq C(1+||y||^{C} )\text{ for some constant }C \},
\end{equation*}
where $b=d$ or $b=d+1$.
Usually, we call elements in $\mathcal{P}_{b}$  polynomially  bounded vectors.

\begin{theorem}\label{Thm.Anotherform}
If the equation
\begin{equation}\label{G.Anotherform}
     v(n)^{-1}\varphi_n +\sum_{\eta\in\mathbb{Z}^d}\lambda\Gamma_0(n-\eta)\varphi_{\eta}=0
\end{equation}
with $n\in \mathbb{Z}^d$ has a nonzero polynomially  bounded solution $\varphi$, then
\begin{equation}\label{G.Anotherform1}
   \psi_{(n,x)}=\sum_{\eta\in\mathbb{Z}^d}G_0(n-\eta,x)\varphi_{\eta}
\end{equation}
solves equation (\ref{Def.H}) and $\psi_{(n,x)}$ is nonzero  and polynomially  bounded.
If equation (\ref{Def.H}) has a nonzero  polynomially  bounded solution,
define
\begin{equation}\label{varphi_n}
    \varphi_n=\sum_{\eta\in\mathbb{Z}^d}\Gamma_0^{-1}(n-\eta) \psi_{( \eta,0)}.
\end{equation}
Then $\varphi_n$ solves equation (\ref{G.Anotherform}) and $\varphi_n$ is nonzero and polynomially  bounded.
\end{theorem}
Notice that for $E\in (-\infty,\infty)\backslash [-c_d,c_d]$, the operator $\Gamma_0$ is well defined and
\begin{equation*}
   |\Gamma_0^{-1}(n )|,  |\Gamma_0(n )|\leq e^{-t||n||}
\end{equation*}
for some $t>0$ ($t$ depends on the distance between $E$ and $[-c_d,c_d]$). Similarly, for some $t>0$,
\begin{equation*}
    |G_0(n,x)|\leq e^{-t(||n||+|x|)}
\end{equation*}
with $n\in \mathbb{Z}^d$.
\par
The proof of Theorem \ref{Thm.Anotherform} can be found in \cite{khoruzhenko1997localization}.
For convenience, we give a short proof in the Appendix.
\par

Now we introduce  the Cayley transformation for a self-adjoint operator.
\begin{definition}\label{Def.Cayley}
  Given a self-adjoint operator $A$ on Hilbert space $ \mathbb{H} $, we call operator
$ ( I-iA)(  I +iA)^{-1}$  the Cayley transformation for operator $A$.
 \end{definition}
Clearly,  $( I+iA)^{-1}$ is well defined and $ ( I-iA)(  I +iA)^{-1}$ is a unitary operator.
\par
\begin{lemma}(\cite{cyconschrodinger,MR776654})\label{TRF1}
If   a vector $u\in  \mathcal{P}_d $ satisfies  equation (\ref{G.Anotherform}), and define $ c=(1+i\lambda\Gamma_0)\varphi$, then   $c\in  \mathcal{P} _d$ and
\begin{equation}\label{Clay}
    (1+iv^{-1})(1-iv^{-1})^{-1}c=(1-i\lambda\Gamma_0)(1+i\lambda\Gamma_0)^{-1}c.
\end{equation}
And the converse is true. More concretely, if $c\in  \mathcal{P} _d$  and satisfies (\ref{Clay}), then
the vector $ \varphi=(1+i\lambda\Gamma_0)^{-1}c \in \mathcal{P}_d$  and satisfies equation  (\ref{G.Anotherform}).
\end{lemma}

It is easy to see that operator $ (1+iv^{-1})(1-iv^{-1})^{-1}$   is multiplication by
\begin{equation*}
   - e^{-2\pi i(n\cdot\alpha+\theta) }
\end{equation*}
 We extend the  Fourier Transform to  a sequence $f_n \in \mathcal{P}_d$  in distributional sense.
More concretely, for given  sequence $f_n \in \mathcal{P}_d$, define the  Fourier Transform by the following equation,
\begin{equation}
\nonumber
    \hat{f}(y)=\sum_{n\in \mathbb{Z}^d} e^{-2\pi in\cdot y}f_n.
\end{equation}
By Fourier transformation,
the  equation (\ref{Clay}) becomes
\begin{equation}\label{TRF3}
    q(y)\hat{c}(y)= e^{ -2\pi i\theta}\hat{c}(y+ \alpha),
\end{equation}
where $ q(y)=-\frac{1-i \lambda\hat{\Gamma}_0(y)}{1+i  \lambda\hat{\Gamma}_0(y)}$.
It is easy to see that $q(y)$ is analytic on $\mathbb{T}^d$, and  $|q(y)|=1$  for $y\in\mathbb{ R}^d$.
We can rewrite
\begin{equation}\label{defzeta}
    q(y)=e^{- 2\pi i\zeta(y)},
\end{equation}
where $\zeta$ is  analytic on $\mathbb{T}^d$. Notice that $\zeta$ only depends on  $\lambda$ and $E$.

Next we will study the ergodicity of the analytical function, the proof is similar to the proof in \cite{liulana1}.

Let $ \frac{p_n}{q_n}$ be the continued fraction approximants to $\alpha\in \mathbb{R}\backslash \mathbb{Q}$. Then
\begin{equation}\label{GDC1}
\forall 1\leq k <q_{n+1},  \text{dist}( k\alpha,\mathbb{Z})\geq  |q_n\alpha -p_n|,
\end{equation}
and
\begin{equation}\label{GDC2}
      \frac{1}{2q_{n+1}}\leq\Delta_n:=|q_n\alpha-p_n| \leq\frac{1}{q_{n+1}}.
\end{equation}
We define an index $\beta(\alpha)$ for $\alpha\in \mathbb{R}\backslash \mathbb{Q}$,
\begin{equation}\label{Def.beta}
    \beta(\alpha)=\limsup_{n\to \infty}\frac{\ln q_{n+1}}{q_n}.
\end{equation}

\begin{lemma}\label{Le}
Let $f: \mathbb{T}=\mathbb{R} / \mathbb{Z}\rightarrow \mathbb{R} $ be a real analytical function on the strip $\{z: |\Im z|\leq\rho\}$ and
$\alpha\in \mathbb{R}\backslash \mathbb{Q}$.
Suppose $\beta(\alpha)=0$. Then for any
integer sequence $\{j_k\}$ such that
\begin{equation*}
   \lim_{k\to \infty} ||j_k\alpha||_{\mathbb{R} / \mathbb{Z}}=0,
\end{equation*}
we have  for $x \in \{z: |\Im z|\leq \frac{\rho}{2}\}$ uniformly,
\begin{equation}\label{uniform1}
   \lim_{k\to \infty} (\sum_{j=0}^{j_k-1} f(x+j\alpha)-j_k\int_{\mathbb{T}}f(x)dx)=0.
\end{equation}
\end{lemma}
\begin{proof}
By the assumption, one has $|f_k|\leq Ce^{-2\pi \rho|k|}$, where $f_k$ is the Fourier coefficients of $f(x)$.
\begin{equation*}
    \sum_{m=0}^{j_k-1} f(x+m\alpha)-j_k\int_{\mathbb{T}}f(x)= \sum_{j\neq 0}\frac{1-e^{-2\pi ijj_k\alpha}}{1-e^{-2\pi ij\alpha}}  {f} _je^{-2\pi jx}.
\end{equation*}
By the fact $ ||j_k \alpha||_{\mathbb{R}/\mathbb{Z}}\to 0$,  it suffices to show

\begin{equation}\label{appendixG12}
     \sum_{j\neq 0}\sup_{k}|\frac{1-e^{-2\pi ijj_k\alpha}}{1-e^{-2\pi ij\alpha}}  {f} _j |e^{ \rho \pi|j|}<\infty.
   \end{equation}

By (\ref{GDC1}),(\ref{GDC2}) and $\beta(\alpha)=0$, one has
\begin{equation*}
   ||j\alpha||_{\mathbb{R} / \mathbb{Z}}\geq c e^{-\frac{\rho}{4}|j|}.
\end{equation*}
Thus
\begin{eqnarray*}
 |\frac{1-e^{-2\pi ijj_k\alpha}}{1-e^{-2\pi ij\alpha}}  {f} _j | &\leq& C e^{\frac{\rho}{4}|j|}e^{-2\pi \rho|j|} \\
   &\leq& C e^{-\frac{3\pi \rho}{2}|j|}
\end{eqnarray*}
This implies (\ref{appendixG12}).
\end{proof}

Now we concern the case $\beta(\alpha)>0$, by the definition, there exists a sequence $\{{q}_{n_k}\}$ such that
\begin{equation}\label{qk}
   q_{n_k+1}\geq  e^{\frac{3}{4}\beta {q}_{n_k}}.
\end{equation}

\begin{lemma}\label{Le1}
Let $f: \mathbb{T}=\mathbb{R} / \mathbb{Z}\rightarrow \mathbb{R} $ be a real analytical function on the strip $\{z: |\Im z|\leq\rho\}$.  Given $\alpha\in \mathbb{R}\backslash \mathbb{Q}$,  suppose $\beta(\alpha)>0$ and let  $\bar{\rho}=\frac{1}{30}\inf\{\rho, \beta\}$.
 Let
\begin{equation*}
\{t_m\}_{m=1}^{\infty}=\{{q}_{n_k}, 2{q}_{n_k}, 3{q}_{n_k},\cdots, \ell_{ k} {q}_{n_k}\} _{k=1}^{\infty} ,
\end{equation*}
where ${q}_{n_k}$ given by (\ref{qk}) and $ \ell_{ k}=\lfloor e^{\bar{\rho} {q}_{n_k}}\rfloor$\footnote{ $\lfloor \ell \rfloor$ denotes  the smallest integer not exceeding $\ell$}.
Then for $x \in \{z: |\Im z|\leq \frac{\bar{\rho}}{2}\}$ uniformly,
\begin{equation}\label{uniform}
   \lim_{m\to \infty} (\sum_{j=0}^{t_m-1} f(x+j\alpha)-t_m\int_{\mathbb{T}}f(x)dx)=0.
\end{equation}

\end{lemma}
\begin{proof}
First, we have
\begin{equation*}
    \sum_{j=0}^{t_m-1} f(x+j\alpha)-t_m\int_{\mathbb{T}}f(x)= \sum_{j\neq 0}\frac{1-e^{-2\pi ijt_m\alpha}}{1-e^{-2\pi ij\alpha}}  {f} _j e^{-2\pi jx}.
\end{equation*}
By (\ref{GDC2}) and (\ref{qk}), one has
 \begin{eqnarray*}
              ||t_m \alpha||_{\mathbb{R}/\mathbb{Z}} &\leq & \lfloor e^{\bar{\rho} {q}_{n_k}}\rfloor || {q}_{n_k}\alpha||_{\mathbb{R}/\mathbb{Z}} \\
               &\leq & Ce^{\bar{\rho} {q}_{n_k}} e^{-3\bar{\rho} {q}_{n_k}}\\
               &\leq & Ce^{-2\bar{\rho} {q}_{n_k}}.
            \end{eqnarray*}
This means
 $ ||t_m \alpha||_{\mathbb{R}/\mathbb{Z}} \to 0$, then it suffices to show
\begin{equation}\label{appendixG1}
     \sum_{j\neq 0}\sup_{m}|\frac{1-e^{-2\pi ijt_m\alpha}}{1-e^{-2\pi ij\alpha}}  {f} _j |e^{{\pi \bar{\rho}} |j|}<\infty.
   \end{equation}
Fix some  $t_m=\ell {q}_{n_k}$ with $0<\ell\leq e^{\bar{\rho} {q}_{n_k}}$.
 \par
 Case 1: $|j|=q_{n}$  for some $n$.

 If  $n\leq  n_k-1$, by (\ref{GDC2}) and (\ref{qk}),
  one has
 \begin{eqnarray*}
    |\frac{1-e^{-2\pi ijt_m\alpha}}{1-e^{-2\pi i j\alpha}}   f  _j | &\leq&  C |j| \ell \frac{ ||{q}_{n_k}\alpha||_{\mathbb{R}/\mathbb{Z}}}{||q_{n}\alpha||_{\mathbb{R}/\mathbb{Z}}}| f  _j| \\
     &\leq&  C |j|  \ell \frac{q_{n+1}}{q_{n_k+1}}| f  _j|\\
      &\leq&      C |j|  e^{\bar{\rho} {q}_{n_k}} \frac{q_{n_k}}{q_{n_k+1}}| f  _j|  \leq e^{-2\pi\bar{\rho}|j|}.
 \end{eqnarray*}
If  $n=  n_k$, by (\ref{GDC2}) and (\ref{qk}),
  one has
 \begin{eqnarray*}
    |\frac{1-e^{-2\pi ijt_m\alpha}}{1-e^{-2\pi i j\alpha}}   f  _j | &\leq&  C |j| \ell \frac{ ||{q}_{n_k}\alpha||_{\mathbb{R}/\mathbb{Z}}}{||q_{n_k}\alpha||_{\mathbb{R}/\mathbb{Z}}}| f  _j| \\
     &\leq&  C |j|  \ell | f  _j|\\
      &\leq&        e^{-2\pi\bar{\rho}|j|}.
 \end{eqnarray*}

 If  $n\geq n_k+1$, by (\ref{GDC2}) again,
  one has
  \begin{eqnarray*}
    |\frac{1-e^{-2\pi ijt_m\alpha}}{1-e^{-2\pi i j\alpha}}   f  _j | &\leq&  C \frac{t_m||j  \alpha||_{\mathbb{R}/\mathbb{Z}}}{||j \alpha||_{\mathbb{R}/\mathbb{Z}}}| f  _j| \\
     &\leq&    e^{-2\pi\bar{\rho}|j|} .
 \end{eqnarray*}

 Case 2:$q_{n}<|j|<q_{n+1}$ for some $n$.

  If   $|j|\geq \frac{1}{\bar{\rho}}\ln q_{n+1}$, by (\ref{GDC1}) and (\ref{GDC2}), we have
  \begin{eqnarray*}
    |\frac{1-e^{-2\pi ijt_m\alpha}}{1-e^{-2\pi i j\alpha}}   f  _j | &\leq&  C  q_{n+1}| f  _j| \\
     &\leq&   C e^{-2\pi\bar{\rho}|j|} .
 \end{eqnarray*}

  If $q_{n}<|j|<q_{n+1}$ and $|j|\leq \frac{1}{\bar{\rho}}\ln q_{n+1}$, let $|j|=\ell^\prime q_{n}+j_0$ with $|\ell^\prime|\leq \frac{\ln q_{n+1}}{q_{n}\bar{\rho}} $
 and $0\leq |j_0|<q_{n^\prime}$.
 Assume $j_0\neq 0$, then by (\ref{GDC1}) and (\ref{GDC2}), one has
   \begin{eqnarray*}
                                          ||j\alpha||_{\mathbb{R}/\mathbb{Z}} &\geq&  \Delta_{n-1}-\ell ^\prime\Delta_{n} \\
                                           &\geq& \frac{1}{C q_n}.
                                       \end{eqnarray*}
Thus we have
 \begin{eqnarray*}
    |\frac{1-e^{-2\pi ijt_m\alpha}}{1-e^{-2\pi i j\alpha}}   f  _j | &\leq& C q_n| f  _j| \\
     &\leq&    e^{-2\pi\bar{\rho}|j|} .
 \end{eqnarray*}
  \par
  Assume $j_0=0$, then  $|j|=\ell^\prime q_{n} $ with $|\ell^\prime|\leq \frac{\ln q_{{n}+1}}{q_{n}\bar{\rho}} $. Applying the
  same proof as in the  Case 1,
  we also have
  \begin{equation*}
   |\frac{1-e^{-2\pi ijt_m\alpha}}{1-e^{-2\pi i j\alpha}}   f  _j |<e^{-2\pi\bar{\rho}|j|} .
  \end{equation*}
  The  estimate (\ref{appendixG1})   is easy to obtain from the above   cases.
\end{proof}
{\bf The proof of Theorem \ref{Thac}}

\begin{proof}
By Theorem \ref{acsp} and Remark \ref{revisere},
it suffices to show  that in regime $\mathbb{R}\backslash[-c_d,c_d]$ ($d=1$), the spectral measures of surface Maryland model are  supported on a Hausdorff dimension zero set(we only prove the case for operator  $H$).
 Let
\begin{equation}\label{Support}
    B=\{E: E\in \mathbb{R}\backslash[-4,4] \text{ such that equation }   H\psi=E\psi\text{ has a polynomial bounded solution } \}.
\end{equation}
 By  Sch'nol Theorem  \cite{berezanskii1968expansions},  $B$ is a support of spectral measures   restriction on  $(-\infty,\infty)\backslash [-4,4]$.
 We will  prove that
$B$ is a  Hausdorff dimension zero set, which implies the Theorem directly.
Below is the details.
\par

Let $E\in B$, then there exists a  polynomial bounded solution of  $\psi_{n,x} $ of  equation $  H\psi=E\psi$.
By Theorem \ref{Thm.Anotherform} and Lemma \ref{TRF1}, we have (using (\ref{TRF3}))
\begin{equation}\label{addTRF3}
   e^{- 2\pi i\zeta(y)}\hat{c}(y)= e^{ -2\pi i\theta}\hat{c}(y+ \alpha),
\end{equation}
where $\hat{c}(y)$ is the Fourier transformation of a polynomially bounded vector $c\in \ell^2(\mathbb{Z})$.

Case 1: $\beta(\alpha)=0$.

For any sequence $\{j_k\}$ such that $ \lim_{k\to \infty} ||j_k\alpha||_{\mathbb{R} / \mathbb{Z}}=0$,
iterate (\ref{addTRF3}) $j_k$ times, we have
\begin{equation*}
   e^{- 2\pi i\sum_{j=0}^{j_k-1}\zeta(y+j\alpha)}\hat{c}(y)= e^{ -2\pi i j_k\theta}\hat{c}(y+ j_k\alpha).
\end{equation*}
Notice that $\hat{c}(y+ j_k\alpha)\to \hat{c}(y)$ and
combine  with (\ref{uniform1}), one has
\begin{equation*}
   \lim_{n\rightarrow\infty} ||j_k (\zeta_0(E)-\theta)||_{\mathbb{R}/\mathbb{Z}}=0,
 \end{equation*}
where  $\zeta_0 =\int_{\mathbb{R}/\mathbb{Z}}\zeta(y) dy $.
By the arbitrary of $j_k$, we must have
\begin{equation*}
    \zeta_0(E)\in \theta+\alpha \mathbb{Z}+\mathbb{Z}.
\end{equation*}
Notice that $\zeta_0(E) $  (non-constance) is analytical for $E\in (-\infty,\infty)\backslash [-4,4]$, then
set $B$ only contains countable set. Which means the Hausdorff dimension of $B$ is zero.

Case 2: $\beta(\alpha)>0$.

By the definition of  $\zeta(y)$,  $\zeta(y)$  is analytical on the strip $\{z:|\Im z|\leq \rho(E)\}$, where   $\rho(E)>0 $ is uniform with respect to
$|E|-4$.
Let $\{t_m\}$ be given by Lemma \ref{Le1}
 and iterate (\ref{addTRF3}) $t_m$ times, we have
\begin{equation*}
   e^{- 2\pi i\sum_{j=0}^{t_m-1}\zeta(y+j\alpha)}\hat{c}(y)= e^{ -2\pi it_m\theta}\hat{c}(y+ t_m\alpha).
\end{equation*}
By (\ref{uniform}) and the fact $ ||t_m \alpha||_{\mathbb{R}/\mathbb{Z}} \to 0$, we get
\begin{equation*}
   \lim_{m\to \infty} ||t_m(\zeta_0(E)-\theta)||_{\mathbb{R}/\mathbb{Z}}=0.
\end{equation*}
This implies for large $k$,
\begin{equation*}
    ||q_{n_k}(\zeta_0(E)-\theta)||_{\mathbb{R}/\mathbb{Z}}\leq   e^{-\frac{\bar{\rho}}{2}q_{n_k} },
\end{equation*}
where $\bar{\rho}$ has a uniformly nonzero lower bound depending on $|E|-4$ and $\beta(\alpha)$.

It is easy to check the set $\{x\in \mathbb{R}: ||q_{n_k}x||_{\mathbb{R}/\mathbb{Z}}\leq   e^{-\frac{\bar{\rho}}{2}q_{n_k} } \text{for large } k\}$
is a  Hausdorff measure zero set.
 In addition $\zeta_0(E) $ is analytical for $E\in (-\infty,\infty)\backslash [-4,4]$, we know $B$ is also a  Hausdorff measure zero set.

 For operator $H^{+}$,  we only need to replace the operator $G_0$, $\Gamma_0$ etc.  with  $G_0^+$, $\Gamma_0^+$ etc.
To avoid   repetition, we  omit the proof.
\end{proof}
\appendix
\section{The proof of Theorem \ref{Thm.Anotherform} }
{\bf The proof of Theorem \ref{Thm.Anotherform}}
\begin{proof}
If  $\varphi_n$ is a nonzero polynomially  bounded and solves equation (\ref{G.Anotherform}),  it is easy to compute directly that
$ \psi_{n,x}$ given by equation (\ref{G.Anotherform1}) is nonzero polynomially  bounded and solves equation (\ref{Def.H}).
\par
To the converse, suppose $ \psi_{n,x}$ is nonzero polynomially  bounded and solves equation (\ref{Def.H}),
this implies
\begin{equation}\label{G.appendix1}
    ((H_0-EI)\psi)_{(n,x)}=-\delta(x)\lambda v(n)\psi_{n,x}.
\end{equation}
Act $G_0=(H_0-EI)^{-1}$ on equation (\ref{G.appendix1}), one has
\begin{equation}\label{G.appendix2}
    \psi_{(n,x)}=- \sum _{\eta\in \mathbb{Z}^d}G_0(n-\eta,x)\lambda v(\eta)\psi_{\eta,0}.
\end{equation}
Thus $\psi_{\eta,0} $ is a non-zero vector, and let $x=0$ in (\ref{G.appendix2}),
we obtain
\begin{equation*}
    \psi_{(n,0)}=-\sum _{\eta\in \mathbb{Z}^d} \Gamma_0(n-\eta) \lambda v(\eta)\psi_{\eta,0}.
\end{equation*}
This yields that $\varphi_n$ given by (\ref{varphi_n}) solves equation (\ref{G.Anotherform}) and is nonzero polynomially bounded. We complete the proof
of Theorem \ref{Thm.Anotherform}.
\end{proof}

 \section*{Acknowledgments}
This research was partially
 supported by NSF DMS-1401204.
 I would like to thank Svetlana Jitomirskaya for drawing my attention to the problem and
inspiring discussions on this subject.


\footnotesize
\begin{thebibliography}{10}

\bibitem{bentosela2008spectral}
F.~Bentosela, P.~Briet, and L.~Pastur.
\newblock On the spectral and wave propagation properties of the surface
  maryland model.
\newblock {\em Journal of Mathematical Physics}, 44(1):1--35, 2003.

\bibitem{bentosela2005spectral}
F.~Bentosela, P.~Briet, and L.~Pastur.
\newblock Spectral analysis of the generalized surface maryland model.
\newblock {\em St. Petersburg Mathematical Journal}, 16(6):923--942, 2005.

\bibitem{berezanskii1968expansions}
J.~M. Berezanskii.
\newblock Expansions in eigenfunctions of self-adjoint operators.
  {T}ranslations of {M}athematical {M}onographs, vol. 17.
\newblock {\em American Mathematical Society, Providence, RI}, 1968.

\bibitem{cyconschrodinger}
H.~Cycon, R.~Froese, W.~Kirsch, and B.~Simon.
\newblock Schr{\"o}dinger operators with applications to quantum mechanics and
  global geometry. 1987.

\bibitem{MR767188}
A.~L. Figotin and L.~A. Pastur.
\newblock An exactly solvable model of a multidimensional incommensurate
  structure.
\newblock {\em Comm. Math. Phys.}, 95(4):401--425, 1984.

\bibitem{grempel1982localization}
D.~Grempel, S.~Fishman, and R.~Prange.
\newblock Localization in an incommensurate potential: An exactly solvable
  model.
\newblock {\em Physical Review Letters}, 49(11):833, 1982.

\bibitem{jakvsic1998spectrum}
V.~Jak{\v{s}}i{\'c} and S.~Molchanov.
\newblock On the spectrum of the surface maryland model.
\newblock {\em Letters in Mathematical Physics}, 45(3):185--193, 1998.

\bibitem{jaksic1998surface}
V.~Jak{\v{s}}i{\'c} and S.~Molchanov.
\newblock On the surface spectrum in dimension two.
\newblock {\em Helvetica Physica Acta}, 71(6):629--657, 1998.

\bibitem{jakvsic1999localization}
V.~Jak{\v{s}}i{\'c} and S.~Molchanov.
\newblock Localization of surface spectra.
\newblock {\em Communications in mathematical physics}, 208(1):153--172, 1999.

\bibitem{liulana1}
S.~Jitomirskaya and W.~Liu.
\newblock {A}rithmetic spectral transitions for the {M}aryland model.
\newblock Preprint,2014.

\bibitem{khoruzhenko1997localization}
B.~A. Khoruzhenko and L.~A. Pastur.
\newblock The localization of surface states: an exactly solvable model.
\newblock {\em Physics Reports}, 288(1):109--126, 1997.

\bibitem{reed1972methods}
M.~Reed and B.~Simon.
\newblock {\em Methods of Modern Mathematical Physics: Vol.: 1.: Functional
  Analysis}.
\newblock Academic press, 1972.

\bibitem{MR776654}
B.~Simon.
\newblock Almost periodic {S}chr\"odinger operators. {IV}. {T}he {M}aryland
  model.
\newblock {\em Ann. Physics}, 159(1):157--183, 1985.

\end{thebibliography}
\end{document}